\documentclass[11pt]{article}
\PassOptionsToPackage{compress, sort}{natbib}
\usepackage[margin = 1in]{geometry}
\usepackage{tikz}
\usetikzlibrary{arrows,automata}
\usetikzlibrary{shapes.geometric,positioning}
\usepackage{fancyhdr}
\usepackage{hyperref}       
\hypersetup{colorlinks,linkcolor=blue,
            citecolor=blue,
            urlcolor=magenta,
            linktocpage,
            plainpages=false}
\usepackage{url}            
\usepackage{amsfonts}       
\usepackage{nicefrac}       
\usepackage{microtype} 
\usepackage{amsmath, amsthm, amssymb}
\usepackage{thm-restate}
\usepackage[ruled,algo2e]{algorithm2e}
\usepackage{algpseudocode}
\usepackage{bbm}
\usepackage{graphicx}
\usepackage{mathtools}
\usepackage{enumitem}
\usepackage{physics}
\usepackage{subcaption}
\usepackage{natbib}
\usepackage{authblk} 
\usepackage{times}
\usepackage{tabulary, booktabs}
\usepackage{thmtools}
\usepackage{thm-restate}

\newcommand{\nosemic}{\SetEndCharOfAlgoLine{\relax}}

\newtheorem{lem}{Lemma}

\newtheorem{rem}{Remark}

\newtheorem{prop}{Proposition}

\newcommand{\Bc}{\mathcal{B}}

\newcommand{\Sc}{\mathcal{S}}
\newcommand{\Dc}{\mathcal{D}}

\newcommand{\Yc}{\mathcal{Y}}
\newcommand{\Ac}{\mathcal{A}}

\newcommand{\E}{\mathbb{E}}

\newcommand{\Prob}{\mathbb{P}}

\DeclareMathOperator*{\argmax}{arg\,max}

\renewenvironment{abstract}
  {{\centering\large\bfseries Abstract\par}\vspace{0.7ex}%
    \bgroup
       \leftskip 20pt\rightskip 20pt\small\noindent\ignorespaces}%
  {\par\egroup\vskip 0.25ex}

{\par\egroup\vskip 0.25ex}

\usepackage{amsmath,amsfonts} 
\usepackage[ruled,algo2e]{algorithm2e}
\usepackage{algpseudocode}
\usepackage{amsthm}
\usepackage{mathtools}
\usepackage{enumitem}
\usepackage[textfont=normal, labelfont=bf]{caption}
\usepackage{subcaption}

\newcommand{\pibase}{\pi_{\text{base}}}
\newcommand{\Ec}{\mathcal{E}}
\newcommand{\Nc}{\mathcal{N}}
\newcommand{\Oc}{\mathcal{O}}
\newcommand{\Tc}{\mathcal{T}}
\newcommand{\Vc}{\mathcal{V}}

\AtBeginDocument{%
  }

\begin{document}

\title{Optimizing Long-term Value for Auction-Based Recommender Systems via On-Policy Reinforcement Learning}

\author{Ruiyang Xu*, Jalaj Bhandari*, Dmytro Korenkevych, Fan Liu, Yuchen He, Alex Nikulkov, Zheqing Zhu}
\affil[]{Meta AI}
\renewcommand\Authands{ and }

\maketitle

\begin{abstract}
Auction-based recommender systems are prevalent in online advertising platforms, but they are typically optimized to allocate recommendation slots based on immediate expected return metrics, neglecting the downstream effects of recommendations on user behavior. In this study, we employ reinforcement learning to optimize for long-term return metrics in an auction-based recommender system. Utilizing temporal difference learning, a fundamental reinforcement learning algorithm, we implement an \textit{one-step policy improvement approach} that biases the system towards recommendations with higher long-term user engagement metrics. This optimizes value over long horizons while maintaining compatibility with the auction framework. Our approach is grounded in dynamic programming ideas which show that our method provably improves upon the existing auction-based base policy. Through an online A/B test conducted on an auction-based recommender system which handles billions of impressions and users daily, we empirically establish that our proposed method outperforms the current production system in terms of long-term user engagement metrics.
\end{abstract}

\def\thefootnote{*}\footnotetext{These authors contributed equally to this work}

\section{Introduction}
With exponential growth of digital information, recommender systems have come to play a pivotal role in various applications, from personalized movie recommendations at Netflix to product recommendation over e-commerce platforms, to help users access relevant/interesting content \citep{gomez2015netflix,smith2017two,lu2015recommender}. 
Given widespread use, research and practice of recommender system design has constantly evolved in the past two decades. These include traditional recommendation strategies, like content based approaches \citep{pazzani2007content,lops2011content} and collaborative filtering methods \citep{schafer2007collaborative, ekstrand2011collaborative, shi2014collaborative, koren2009matrix} which model user preferences for different items either using item features or by using past ratings of a user (or similar users), to modern deep supervised learning based approaches which personalize user experience by modeling click through rates \citep{zhang2019deep}, with many instances of successful deployment at industry scale. See \citep{covington2016deep,cheng2016wide,okura2017embedding} for example of prototypes used for video recommendation at YouTube, app recommendations at Google Play and news recommender systems for Yahoo!.

Auction based systems \citep{varian2007position,evans2008economics} is a critical component for platforms like Meta, Google, Yahoo! etc. which run large marketplaces such as online advertising platforms which allocate ad slots based on advertiser bids. These systems typically use a Vickrey (second price) auction or a Vickrey-Clarke-Groves (VCG) auction\footnote{The VCG auction encompasses the traditional Vickrey auction as a special case.} where the winner pays the next highest bidder's bid. An attractive property of the VCG auction is that bidding the true value is a dominant strategy\footnote{Note that for multiple slots, platforms typically rely on generalized second price (GSP) auctions which truth-telling is not a dominant strategy \citep{edelman2007internet, varian2007position}. For our exposition, we assume a single slot with allocation done using a second price auction.} for all bidders \citep{edelman2007internet, varian2014vcg}, which simplifies bidder's decisions. That is, each participant is incentivised to bid what they value the item to be worth, while the winner gets a discount. 
This has led to widespread adoption of the second price auction in large scale online ad marketplaces \citep{varian2014vcg}.

In addition to bid values, auction based recommender systems also typically account for ``conversion rates'' to allocate slots to recommendations. However, both these metrics only account for the immediate value of showing a recommendation in a slot. In this way, auction based recommender systems are designed to optimize myopically for short-term engagement metrics and may not reflect the impact on long-run user engagement. Thus, optimizing for the long-term in an auction based system poses a unique challenge due to the auction mechanism. In this work, we propose a reinforcement learning (RL) approach to bias an auction based recommender system towards strategies that account for downstream impact of recommendations on user behavior. 

Reinforcement learning provides a mathematical formalism to optimize for long horizon outcomes and has gained traction through superior performance in different applications, including arcade games \citep{mnih2015human}, robotics \citep{gu2017deep, smith2017two} and navigation \citep{tai2017virtual}. Many recent works have also shown potential of using RL methods for real-time recommender systems, including off-policy methods like Q-learning \citep{zou2019reinforcement,zheng2018drn} and actor-critic algorithms \citep{chen2019top, chen2022off}. However, balancing between bidder's value and long-term user engagement is one of the key constraints for system designers in applying RL to an auction based recommender system. Unlike most of the prior work, we cannot design a system from scratch using RL; instead we must work within a framework which continues to optimize for bidder's value as well, at least partially. Moreover, using off-policy RL algorithms to search for an optimal policy poses significant challenges due to distribution shift. Offline collected data rarely satisfies coverage conditions required for off-policy RL algorithms to find a (near) optimal policy. In this work, we make the following contributions:
\begin{itemize}
    \item We use a reinforcement learning approach to bias an auction based recommender system toward recommendations that improve metrics for long-term user engagement. Our approach can be understood as \textit{one step policy improvement} over the auction based recommender policy.
    \item Using ideas from classic dynamic programming theory, we make a simple argument to show that our method provably improves over the base (auction) policy.
    \item We implement our method in an industrial scale real-time recommender system serving billions of users daily and empirically show performance improvements through online A/B testing.
\end{itemize}

This paper is organized as follows. In section \ref{sec: lit review}, we briefly review prior work on recommender systems with a focus on reinforcement learning based approaches. Section \ref{sec: problem formulation} instantiates a reinforcement learning setup to optimize for long-term value in a recommender system problem setting we consider. In section \ref{sec: pi approach}, we describe our online reinforcement learning approach which provably improves over the baseline auction based recommender systems in terms of long-run user engagement metrics (we use ``conversions'' as a proxy metric of long-run user engagement). We end the paper in section \ref{sec: experiments} with experimental results from an online A/B test showing significant improvement in conversion statistics over an six week period.

\begin{rem}
Our approach is not specifically tied to the surrogate metrics we use to quantify long-term user engagement. The surrogate metrics should be looked at  through the lens of reward design.
\end{rem}

\section{Related Work}
\label{sec: lit review}
Many of the existing approaches to personalized recommendations, like content and collaborative filtering based methods \citep{lops2011content,ekstrand2011collaborative,shi2014collaborative}, sequential recommender systems \citep{hidasi2015session,quadrana2018sequence,wang2019sequential} and deep supervised learning based methods \citep{hidasi2015session,zhang2019deep}, optimize metrics of immediate user engagement (e.g. click probabilities) without incorporating the downstream effect of recommendation policies on metrics of long-run user engagement. This includes (contextual) multi-arm bandit approaches \citep{li2010contextual}, although bandit approaches allow for learning user preferences by intelligent exploration thereby combating the feedback loops from data generated by existing recommender systems \citep{wang2017factorization,qin2014contextual,zeng2016online, zhu2023scalable, Guo2023evaluate}.
Given this and the sequential nature of user interactions, reinforcement learning techniques have emerged as natural alternatives to grapple with the problem of myopic recommendation strategies more generally \citep{shani2005mdp,chen2019top,ie2019slateq,xin2020self,zhao2019deep}. Since RL methods aim to optimize for long-run ``value'' (estimate of accrued reward over a long time horizon), the resulting recommendation strategies account for the efficacy of downstream recommendations in recurring user interactions. Although many different approaches have been proposed, including model-based RL methods \citep{chen2019generative}, a lot of recent work has focused on applying model free RL algorithms to maximize long term measures of user engagement in different applications; see \citep{zou2019reinforcement, zheng2018drn, zhu2023deep} for examples in e-commerce setting and a personalizing news article recommendation setting. Similar to these efforts, we train a deep neural network model to approximate a \textit{Q-function} of an auction based recommender policy. However, our focus is on incorporating long-term metrics in an auction-based recommender system. For this, we restrict ourselves to biasing the recommendation policy toward actions with higher Q-values, as opposed to searching for an optimal policy. We use simple ideas from classical dynamic programming theory to elucidate the theoretical motivation of our \textit{one-step policy improvement approach} -- modifying the base policy in this way leads to a policy with improved long-run user engagement metrics.

\section{Problem Formulation}
\label{sec: problem formulation}
In this section, we formalize the long-term value optimization problem of a recommender system in the framework of reinforcement learning. We later describe many approximations we make for practical implementation of a large scale auction based system that is designed around this framework and discuss the limitations of our choices. A general reinforcement learning framework can be characterized by an agent and its interactions with the environment \cite{lu2021reinforcement, hutter2007universal}


\subsection{A model of the environment.}
Following \cite{lu2021reinforcement}, we model the environment by a tuple $\Ec = (\Oc, \Ac, \rho)$, where $\Oc$ and $\Ac$ denote the observation set and action set respectively and $\rho$ prescribes a probability for any observation $o \in \Oc$. For simplicity, we assume that time is discrete and take one time period to denote a 24 hour window for the model we implement. We assume that the set of users interacting with the recommender system is dynamic and denote $\Nc(t)$ to be the set of interacting users at time $t$. To avoid overload of notation, we take the set of all users to be fixed, denoted by $\Nc$; our framework as well as implementation can easily account for new users.

\begin{enumerate}
    \item Action space: We take $\Ac$ to be a fixed\footnote{A fixed action space is assumed to only simplify the exposition in terms of notation.} finite set of all possible items that can be recommended to any user at a given time and let $a \in \Ac$ denote one such recommendation. At time step $t$, an action set $A_t \in \Ac^{|\Nc(t)|}$ is a collection of recommended items to different users, $A_t = \left( a^u_t: u \in \Nc(t) \right)$.
    \item Observation space: We take $\Oc$ to be a finite set of observations about a user, with each observation $o \in \Oc$ encoding useful information about user tastes. The observation set, $O_t \in \Oc^{|\Nc(t)|}$ is a collection of such observations about each user interacting with the system, $O_t = (o^u_t: u \in \Nc(t))$.
    \item Outcome space: We take $\Yc$ to be a finite set of outcomes, with each outcome $y \in \Yc$ indicating a user's response to a recommendation. Similar to the action and observation sets, the outcome set, $Y_t = (y^u_t: u \in \Nc(t)) \in \Yc^{|\Nc(t)|}$ denotes the collection of outcomes for each user interacting with the system at time $t$. We later describe how $Y_t$ depends on various quantities, including the action and observation sets, $A_t$ and $O_t$.
    \item Observation probability: We let history $H_t = (O_0, A_0, Y_0, \ldots, A_{t-1}, Y_{t-1}, O_t)$ denote interactions with the environment up to time $t$. An agent selects actions $A_t$ given the history $H_t$. The probability distribution $\rho(\cdot | H_t, A_t)$ determines the next observation conditioned on $H_t$ and $A_t$. While we leave the exact form of $\rho$ unspecified, $\rho(\cdot)$ essentially models the evolution of user tastes as well as their propensity to interact with the system. 
\end{enumerate}

\subsection{A model of user behavior}
\label{subsec: user behavior}
Given the environment interface as described above, we model the user response $y^u_t$ to a recommended item $a^u_t$ at time $t$ as a function of its past interactions with the system\footnote{History of a user can be written as $h^u_t = \{(o^u_k, a^u_k): k \in \Tc(u)\}$ where $\Tc(u)$ is the set of interaction times of user u with the system up to time $t$. This can be derived from $H_t$ which logs the collection of observation, action sets up to time $t$.} $H_t$, the recommended item $a^u_t$, user context $x^u_t$ which encodes demographic information about the user, as well as some side information $i^u_t$ which encodes dynamic user interests summarized by their interaction with the platform, beyond engagement with the recommender system\footnote{For example, this side information can be derived from organic posts on the platform that users engage with.}. Formally, one can take $g(\cdot)$ to be some unknown fixed function such that, $y^u_t = g(H_t, a^u_t, x^u_t, i^u_t, \epsilon_t)$
where $\epsilon_t$ is i.i.d noise capturing idiosyncratic randomness in user behavior. Throughout we assume that the user contexts as well as side information $(x_u, i_u)$ for each user $u \in \Nc$ evolves independently of other users with time. Essentially, this translates to assuming that there are no \textit{network} effects.

\paragraph{User state representation:} Given our model of the user behavior, we take a user state to include all information about a user available to the recommender system at the time period. In particular, we assume vector encoding $z^u_t = f(H_t, u)$ to summarize a user's interests based on their past interactions with the recommender system. Taken together with the user context vector $x^u_t$ and the side information encoding $i^u_t$, the user state can be parsimoniously represented as the tuple $s^u_t = (z^u_t, x^u_t, i^u_t)$. It is noteworthy that assuming no network effects, the user states are independent of each other, i.e. $s^{u_1} \perp s^{u_2}$ for any pair of users $(u_1, u_2)$. Put differently, a user's state only depends on its past interactions with the recommender system - their context and interests evolve exogenously. Throughout the paper, we denote $s$ to be any generic user state and $\Sc$ to denote the state space (collection of all users states). For simplicity of exposition, we will assume $\Sc$ to be countable.

\subsection{Agent design.}
Besides the environment, a key component of an RL framework is an agent which takes actions given observations from the environment. We let $\pi(\cdot|H_t)$ denote a \textit{policy}, a probability distribution over actions that depends on the history $H_t$. 
\begin{enumerate}
    \item Agent state: From a computational viewpoint, the dependence of a policy on the entire history is problematic. We therefore take $S_t = \{ s^u_t: u \in N(t) \}$ to represent agent ``state'', as a collection of state representations for all users interacting with the system at time $t$. Note that the agent state $S_t$ encodes all relevant information in history $H_t$, including the observation $O_t$. 
    We can write the relationship between observation, states and history as:
    \[
        H_t \to S_t \to O_{t+1} \to S_{t+1} \,\, \text{or equivalently}, \,\, H_{t-1} \to S_{t-1} \to O_t \to S_t \to O_{t+1} \to S_{t+1}.
    \]
    We use the notation $X \to Y \to Z$, which is standard in information theory, to denote that random variables $X$ and $Z$ are independent conditioned on Y. Clearly, the state variable obeys the Markov property, 
    \[
    \Prob(S_{t+1} = \cdot \, | \, S_t, A_t, \cdots, S_0, A_0) = \Prob(S_{t+1} = \cdot \, | \, S_t, A_t),
    \]
    since conditioned on $S_t$ and $A_t$, $S_{t+1}$ is independent of the history $H_t$ and no network effects are assumed. 
    \item Reward function: We denote $r: \Ac \times \Yc \mapsto \{0, 1\}$ to be a reward function which maps an (action, outcome) pair to a scalar, encoding the agent's preferences about outcomes. In our setting, rewards are taken to be binary and $r(a,y) = 1$ indicates ``conversion''. Essentially, we associate an ``end behavior'' for each recommended item $a$. For example, buying might be the end behavior for some recommendations while getting users to subscribe might be the end behavior for others. User behavior (as indicated by the outcome variable) corresponding to the end behavior for an ad results in a reward of 1; all other user behaviors incur a reward of 0. For action, observation set $(A_t, Y_t)$, we assume rewards of individual users are additive and take $R(A_t, Y_t) = \sum_{u \in \Nc(t)} r(a^u_t, y^u_t)$ to be the scalar reward. 
    \item Constraint function: It is common for recommender systems to have constraints on which items can be shown to a user at any given time. We denote $c: \Ac \times \Sc \mapsto \bar{\Ac}$ with $\bar{\Ac} \subseteq \Ac$ to be such a constraint function which restricts the feasible action set to a subset of all possible actions, augmented with an option of ``no recommendation''.
\end{enumerate}

\subsection{Reinforcement learning of recommendation policies}
While the reward function specifies the agent's preferences over a single time step, the goal in reinforcement learning is to optimize for cumulative rewards over long horizons. As is typical, we take a discounting approach to model cumulative rewards with the goal to maximize expected cumulative discounted return,
\begin{equation}
\label{eq: expected cumulative discounted return}
\E \Big[ \sum_{t=0}^{\infty} \gamma^t R(A_t, Y_t) \Big]
\end{equation}
Thus, a good recommendation policy not only optimizes for immediate rewards but also for long run conversions by transitioning users to states with a higher propensity to convert.

As the agent state encodes all relevant information about the history, we denote the agent policy $\pi_{\text{agent}}(\cdot | S)$ to be a distribution over actions given the agent state. Furthermore, our assumption on no network effects (independence of user states across time) and additive reward functions lets us decompose the policy as
\[
\pi_{\text{agent}}(\cdot \, | \, S_t) = \prod_{u \in \Nc(t)} \pi(\cdot \, | \, s^u_t) 
\]
where $\pi(\cdot | s)$ indicates the action selection policy given any generic user state $s$. 
To compare different policies as well as to find the optimal one, reinforcement learning algorithms use value functions. Formally, the value function for policy $\pi$ can be defined as
\begin{equation}
\label{eq: value function}
V_{\pi}(s) = \E_{\pi} \Big[ \sum_{t=0}^{\infty} \gamma^t r(a_t, y_t) \, | \, s_0 = s \Big]
\end{equation}
which measures the cumulative reward under policy $\pi$, from a given user state $s$. Similarly, we can also define the state-action value function as the expected return under policy $\pi$ from a user state $s$ and recommendation $a$,
\begin{equation}
\label{eq: Q funtion}
    Q_{\pi}(s, a) = \E_{\pi} \Big[\sum_{t=0}^{\infty} \gamma^t r(a_t, y_t) \, | \, a_0=a, s_0=s \Big]
\end{equation}
State-action value functions obey a fixed point equation known as the Bellman equation, and can equivalently be expressed as,
\begin{equation}
    \label{eq: bellman equation}
    Q_{\pi}(s,a) = r(a,y) + \gamma \sum_{s' \in \Sc} P(s'|s,a) V_{\pi}(s')
\end{equation}
where $\Sc$ denotes the set of all states and $\{s'\} \in \Sc$ are the set of successor states from the state-action pair $(s,a)$. An optimal policy $\pi^*$, is defined to be the one which maximizes the value-to-go from each state, $\pi^*(s) = \argmax_{\pi} V_{\pi}(s)$. The ultimate goal in RL is to find the optimal policy which maximizes the expected cumulative return from any user state. However, we take a less ambitious approach of performing a one-step policy improvement over a given base policy. We outline this below.

\section{A policy improvement approach}
\label{sec: pi approach}
One of our motivation to carefully understand performance of a policy obtained by a single policy improvement step is tied to business constraints. Auction-based recommender systems are the main revenue generating work-streams of advertisement platforms and therefore it is reasonable to not deviate too much from the given base policy without exhaustive experimentation. Moreover, we are also constrained by working with real user interaction data - we do not have access to a simulator of user behavior that we can repeatedly query to assess the performance of a policy trained  with reinforcement learning. Essentially, from our logged data $\{(O_t, A_t, Y_t)\}_{t \in \mathbb{N}}$, we construct trajectories for each user,
\[
\Dc_{\pibase} = \{(o^u_t, s^u_t, a^u_t, r^u_t)\}_{u \in \Nc, \, t \in \Tc(u)}
\]
where $\Tc(u)$ denotes the set of interaction times of user $u$ with the system, accounting for the fact that not all users choose to interact with the system at all times. Here, the subscript $\pibase$ emphasizes that this interaction data is collected using the (current) base policy. We give details about the base policy in the subsection below.

A standard policy iteration step involves approximating the Q-function of the base policy, $Q_{\pibase}$ using a RL algorithm and subsequently solving the single period optimization problem,
\begin{equation}
\label{eq: policy improvement}
    \pi^{+}(s) = \argmax_{a \in \Ac} \, Q_{\pibase}(s, a)
\end{equation}
Notice from \eqref{eq: bellman equation} that for any policy $\pi$, $Q_{\pi}$ is parameterized by the immediate action which optimizes for the one step reward $r(a,y)$ and transition to the next state $s'$. Therefore, the optimization problem in \eqref{eq: policy improvement} is sometimes referred to as a ``single period'' optimization problem. A well known result in dynamic programming shows that a policy $\pi^{+}$ which simultaneously solves the single step policy iteration problem for all states is an improved policy \citep{bertsekas2012dynamic}. That is,
\[
V_{\pi^{+}}(s) \geq V_{\pibase}(s) \,\,\, \forall \, s \in \Sc
\]
For a problem with finite state and action spaces, a series of policy improvement steps ultimately lead to an optimal policy geometrically fast \citep{bertsekas2012dynamic}. While we can technically obtain an improved policy over the base recommender policy using \eqref{eq: policy improvement}, we are bound by business constraints of staying close to the base policy. Nevertheless, we show that it is possible to find an improved policy by taking a step in the policy improvement direction.

\subsection{Baseline recommender policy with auction mechanism}
The base policy, $\pibase$ that we consider is a second price auction based policy. Essentially, for a given user state $s$, bidders provide a ``bid'' value $\text{Bid}(s,a)$ for a subset of eligible recommendations $a \in \Ac^{b} \subset \bar{\Ac}$ which represents an estimate of their utility of recommending item $a$ in user state $s$. As is common in many practical instantiations of recommender systems, an expected conversion rate (``eCVR'') factor is multiplied to the bid values to upper bound the expected revenue of showing an ad in user state $s$. We take $f(s,a) = \text{Bid}(s,a) \times \text{eCVR}(s,a)$ to denote the ``bid-eCVR'' values. Without going too much into details, one can take the deterministic function $f(\cdot)$ to be a black-box recommendation scoring model. Note however that this recommendation scoring model is not trained for optimizing conversions over long time horizons; rather it is trained to be a good predictor of short run engagement of the users with different items (it takes immediate conversion rates as input). Another common practice in auction-based recommender systems is that each bidder typically only bids on a unique subset of recommendations associated with campaigns. Therefore, it is reasonable to assume that $\Ac^{b_1} \cap \Ac^{b_2} = \emptyset$ for any two bidders $b_1$ and $b_2$. 

In this work, we take the agent's base policy $\pibase$ to be greedy with respect to the bid-eCVR values,
\begin{align}
\label{eq: greedy base policy}
    \pibase(s) = \argmax_{a \in \bar{\Ac}} \, f(s,a) 
\end{align}
where $\bar{\Ac}$ is the set of all eligible recommendations. Note that our description of the base policy in \eqref{eq: greedy base policy} is not arbitrary; it is a good approximation to many practical auction based recommender systems. The auction mechanism is critical to recommender systems and is commonly used across various monetization products and online marketplaces. Moreover, platforms invest a lot of engineering resources to build, maintain and fine-tune conversion rate models. This limits us from re-inventing the wheel from scratch. Instead, we treat $f(\cdot)$ as a fixed black-box input to the baseline policy and aim to find an improved policy (in terms of cumulative ``conversions'' as defined above) by perturbing $f(\cdot)$ with an estimate of the expected long run conversions of recommending item $a$ in user state $s$, $Q_{\pibase}(s,a)$.

\subsection{Optimizing for long-run conversions: a one-step improved policy} \label{ltv_weighted_sum}
Our proposed solution is to bias the recommendation scoring model $f$ to account for long run conversions. Specifically, let $Q_{\pibase}$ denote the state-action value function of the base policy, with $Q_{\pibase}(s,a)$ denoting the expected cumulative conversions attributable to recommending item $a$ in user state $s$ and following the base policy thereafter. We take a modified policy to recommend items based on a weighted score, which is a convex combination of the recommendation score and estimated long run conversions,
\begin{equation}
    \label{eq: modified policy}
    \pi_{\text{mod}}(s) = \text{SELECT}_{a \in \bar{\Ac}} \, \left( f(s,a), \hat{Q}_{\pibase}(s,a) \right)
\end{equation}
where we define the selection operator as 
\begin{align*}
\text{SELECT}_{a \in \bar{\Ac}} \, &\left(f(s,a), \hat{Q}_{\pibase}(s,a)\right) := \argmax_{a \in \bar{\Ac}} \, \left[(1-\alpha)\cdot f(s,a) + \alpha \cdot \hat{Q}_{\pibase}(s,a)\right]
\end{align*}
for a fixed value of $\alpha$. See figure \ref{fig: RL_mod_auction} for an illustration of our approach. Note that when $\alpha=1$, $\bar{\Ac} = \Ac$ and exact Q values are available,
\[
\pi_{\text{mod}}(s) = \pi^{+}(s) = \argmax_{a \in \Ac} \, Q_{\pibase}(s, a)
\]
and the modified policy $\pi_{\text{mod}}$ mimics the policy iteration step. When $\alpha=0$, $\pi_{\text{mod}}$ equals the base policy which is greedy with respect to the recommendation scoring model $f$. Values of $\alpha \in (0, 1)$ interpolate between the two policies, possibly in a non-linear way and fixing a value of $\alpha$ pins down the trade-off between optimizing short run engagement versus optimizing for long run conversions. In practice, we use a value of $\alpha=0.96$ which was chosen empirically after a few A/B test iterations, taking into account the business constraint of staying close to the base policy.
We discuss more about the choice of $\alpha$ in section \ref{sec: experiments}.

\begin{figure}[ht]
  \centering
  \includegraphics[width=0.75\linewidth]{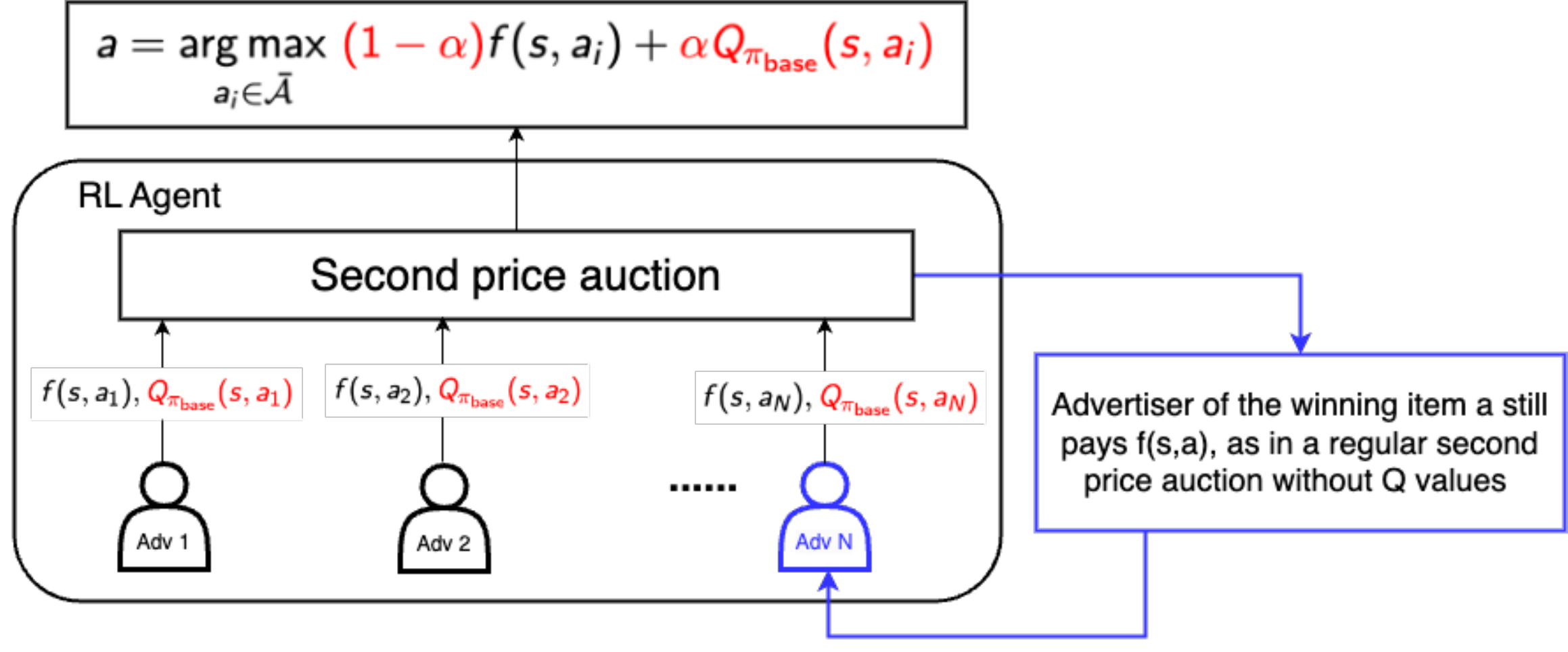}
  \caption{Visualization of an RL approach to improve auction-based recommender policies. We optimize a weighted combination of recommender score $f(s,a)$, which is used for item selection in a second price auction, and estimated long run conversions $\hat{Q}_{\pibase}(s,a)$.}
    \label{fig: RL_mod_auction}
\end{figure}


Next, we show how the modified policy in \eqref{eq: modified policy} with exact Q values yields an improved policy, $V_{\pi_{\text{mod}}}(s) \geq V_{\pibase}(s)$ for all user states $s$. That is, the expected cumulative conversions under the modified policy from each user state improves as compared to the base policy, with possible strict improvements from a subset of user states. This result is not unexpected as the base policy is not optimized for long run conversions. In comparison, the modified policy takes into account expected cumulative conversions over long horizons, at least partially. Nevertheless, this result acts as a sanity check and serves as a motivation for our approach. Recall the definition of value function from \eqref{eq: value function} and the functional form of the modified policy given in \eqref{eq: modified policy}. For our results, we assume that the function $f$ takes unique values for different actions in each user state, i.e. $f(s,a_i) \neq f(s,a_j)$ for any $(a_i, a_j) \in \bar{\Ac}$. This is reasonable as we expect bid-eCVR values for different recommendations to be non-unique. In addition, we also assume that we have access to exact Q values of the base policy, $Q_{\pibase}$.

\begin{prop}[Policy improvement step]
\label{prop: policy improvement}
For any user state $s$, $V_{\pi_{\text{mod}}}(s) \geq V_{\pibase}(s)$. Moreover, $V_{\pi_{\text{mod}}}(s) > V_{\pibase}(s)$ if $\pi_{\text{mod}}(s) \neq \pibase(s)$. That is, the inequality is strict if recommendations under the modified and base policies differ in a user state. 
\end{prop}

The proof follows by a straightforward application of the monotonicity property of the Bellman operators \citep{bertsekas2012dynamic} along with the following lemma. Note that $Q_{\pibase}(s, \pi_{\text{mod}}(s))$ denotes the expected cumulative conversions using the modified policy in user state $s$ and following the base policy thereafter.

\begin{lem}
\label{lem: action improv}
For any user state $s$, $Q_{\pibase}(s, \pi_{\text{mod}}(s)) \geq V_{\pibase}(s)$.
\end{lem}
\begin{proof}
    Consider a fixed state $s$ and let $a_{\text{base}}$, $a_{\text{mod}}$ denote the selected actions under the base and modified policy, with the same eligible\footnote{While comparing action selections under the base and modified policy, we assume the eligible action set remains the same. This seems reasonable to compare the outcome of two policies. Note that the constraint function (which specifies the eligible item set) is part of agent design and is separate from learning of recommendation policies.} item set $\bar{\Ac}$
    \begin{align}
        a_{\text{base}} &= \pibase(s) = \argmax_{a \in \bar{\Ac}} \, f(s,a) \label{eq: base action} \\ 
        a_{\text{mod}} &= \pi_{\text{mod}}(s) = \argmax_{a \in \bar{\Ac}} \, \left[ (1-\alpha)f(s,a) + \alpha Q_{\pibase}(s,a) \right] \label{eq: modified action}
    \end{align}
    Then, we show that, $Q_{\pibase}(s, a_{\text{mod}}) \geq Q_{\pibase}(s, a_{\text{base}})$. Clearly, this holds when $a_{\text{mod}} = a_{\text{base}}$. Let's look at the case when the two actions are different. In this case, first note that $f(s, a_{\text{base}}) > f(s, a_{\text{mod}})$ using the fact that $a_{\text{base}} = \argmax_{a \in \bar{\Ac}} \, f(s,a)$ and that $f(\cdot)$ takes on unique values. However, from \eqref{eq: modified action}, it holds that
    \begin{align*}
        (1-\alpha)f(s,a_{\text{mod}}) + \alpha Q_{\pibase}&(s,a_{\text{mod}}) \geq (1-\alpha)f(s,a_{\text{base}}) + \alpha Q_{\pibase}(s,a_{\text{base}}).
    \end{align*}
    Rearranging, we get that 
    \begin{align*}
        \alpha ( Q_{\pibase}(s, a_{\text{mod}}) - Q_{\pibase}&(s, a_{\text{base}}) ) \geq (1-\alpha) ( f(s,a_{\text{base}}) - f(s,a_{\text{mod}}) ) > 0,
    \end{align*}
    where the last inequality follows as $f(s,a_{\text{base}}) > f(s,a_{\text{mod}})$. Thus, for the case when $a_{\text{mod}}\neq a_{\text{base}}$, we have 
    \begin{equation}
    \label{eq: higher q function under modified policy}
    Q_{\pibase}(s, a_{\text{mod}}) > Q_{\pibase}(s, a_{\text{base}}).
    \end{equation}
    Combining both the cases, we conclude that $Q_{\pibase}(s,a_{\text{mod}}) \geq Q_{\pibase}(s,a_{\text{base}})$. Note that as both $\pibase$ and $\pi_{\text{mod}}$ are deterministic policies, we get our final result,
    \begin{align*}
        Q_{\pibase}(s, \pi_{\text{mod}}(s)) = Q_{\pibase}(s, a_{\text{mod}}) \geq Q_{\pibase}(s, a_{\text{base}}) = Q_{\pi_{\text{base}}}(s, \pi_{\text{base}}(s)) = V_{\pibase}(s).
    \end{align*}
    Here the final equality uses that $Q_{\pi}(s, \pi(s)) = V_{\pi}(s)$. This relationship is well known in RL (see for example \citep{bertsekas2012dynamic}) and can also be derived from equations \eqref{eq: value function} and \eqref{eq: bellman equation} above.
    Note that equation \eqref{eq: higher q function under modified policy} implies a strict inequality, $Q_{\pibase}(s, \pi_{\text{mod}}(s)) > V_{\pibase}(s)$, when the modified policy makes a different recommendation than the base policy.
\end{proof}

\paragraph{Proof of Proposition \ref{prop: policy improvement}:}
We use the result in Lemma \ref{lem: action improv} to show how the modified policy gives an improvement over the base policy. This essentially follows from using the standard monotonicity properties of the Bellman operator \citep{bertsekas2012dynamic}. Note that as per step rewards are bounded in our setting, $|r(a,y)| \leq 1$, value functions corresponding to any stationary policy $\pi$ are also bounded. Let $\Vc = \{V_{\pi}: \pi \in \Pi\}$ be the set of all bounded value functions corresponding to stationary policies $\pi \in \Pi$. The bellman operator $T_{\mu}: \Vc \mapsto \Vc$ for some policy $\mu \in \Pi$ maps a value functions $V \in \Vc$ as 
\[
T_{\mu} V(s) := r(\mu(s), y) + \gamma \sum_{s' \in \Sc} P(s'|s, \mu(s)) V(s') \,\,\,\, \forall \, s \in \Sc.
\]
Using this and the definition of Q function in Equation \eqref{eq: bellman equation}, it is easy to check that for any two stationary policies $\pi, \pi'$,
\begin{equation}
\label{eq: q function with bellman op}
    Q_{\pi}(s, \pi'(s)) = T_{\pi'}V_{\pi}(s)
\end{equation}
With bounded per-step rewards (our setting), Bellman operators are monotone \citep{bertsekas2012dynamic}, implying that if $V_{\pi}(s) \geq V_{\pi'}(s)$ for all $s \in \Sc$, then
\begin{equation}
\label{eq: monotonicity bellman op}
    T_{\mu} V_{\pi}(s) \geq T_{\mu} V_{\pi'}(s) \quad \forall \, s \in \Sc.
\end{equation}
Another standard result from dynamic programming theory \citep{bertsekas2012dynamic} helps us prove proposition \ref{prop: policy improvement}.
\begin{equation}
\label{eq: value function as limit of bellman op}
    V_{\mu} = \lim_{k \to \infty}\, T^{k}_{\mu} V_{\pi}
\end{equation}
Essentially, this states that repeated application of the Bellman operator $T_{\mu}(\cdot)$ converges to its corresponding value function $V_{\mu}$. Using \eqref{eq: q function with bellman op}, we can express the result of lemma \ref{lem: action improv} as
\begin{equation}
\label{eq: rewrite lemma 1}
    T_{\pi_{\text{mod}}}V_{\pibase}(s) \geq V_{\pibase}(s) \quad \forall \, s \in \Sc.
\end{equation}
Applying the monotonicity property to \eqref{eq: rewrite lemma 1} gives,
\begin{equation*}
    T^2_{\pi_{\text{mod}}}V_{\pibase}(s) \geq T_{\pi_{\text{mod}}}V_{\pibase}(s) \geq V_{\pibase}(s)
\end{equation*}
Repeatedly applying the monotonicity property shows the desired result,
\begin{align*}
    V_{\pi_{\text{mod}}}(s) = \lim_{k \to \infty} T^k_{\pi_{\text{mod}}} V_{\pibase}(s) \geq \ldots \geq T^2_{\pi_{\text{mod}}}V_{\pibase}(s)
    \geq T_{\pi_{\text{mod}}}V_{\pibase}(s) \geq V_{\pibase}(s).
\end{align*}
It is noteworthy that the inequality is strict for user states where the recommendation actions under $\pi_{\text{mod}}$ differ from those under $\pibase$. Indeed, proof of lemma \ref{lem: action improv} shows that the inequality in \eqref{eq: rewrite lemma 1} is strict for these states.

\subsection{Implementation using SARSA}
To estimate $Q_{\pibase}$, we train a deep neural network based model using SARSA, a popular \textit{on-policy} reinforcement learning algorithm, and logged data of user interactions  $\{(s^u_t, a^u_t, r^u_t)\}_{u \in \Nc, \, t \in \Tc(u)}$. We pre-process the trajectory data of each user into tuples, 
\[
d^u_{t_i} = \left(s^u_{t_i}, a^u_{t_i}, r^u_{t_i}, s^u_{t_{i+1}}, a^u_{t_{i+1}}, \tau_i\right) \,\,\, \forall \, t_i \in \Tc(u) = \left\{t_0, t_1, \ldots, t_{|\Tc(u)|}\right\}
\]
Recall that $\Tc(u)$ denote the set of interaction times of a user $u$ with the recommender system. 
We let $\tau_i= t_{i+1} - t_i$ denote the $i^{\text{th}}$ inter-interaction time and append it to each tuple. Reinforcement learning algorithms like SARSA which do credit assignment by bootstrapping only require training data in the form of these data tuples, unlike Monte Carlo methods which require entire user interaction trajectories. SARSA is a widely used standard reinforcement learning algorithm -- however, for completeness we give some details in Algorithm \ref{alg: SARSA} below.

Let $\Dc = \{d^u_{t_i}\}_{u \in \Nc, \, t_i \in \Tc(u)}$ be the collection of all user interaction tuples. In reinforcement learning, this is usually referred to as the replay buffer. At each training epoch, data tuples are sampled from $\Dc$ to train the SARSA algorithm. 
Note that the user interaction data is collected under the base policy, $\pibase$ and so we can hope to estimate $Q_{\pibase}$ using an \textit{on-policy} algorithm like SARSA. We also train our models \textit{online}, instead of using a fixed offline data set for training, by updating data set $\Dc$ daily with a fresh set of user interaction data. See section \ref{subsec: data processing} for more details on how we construct these data tuples from user trajectories in an online manner. User states which do not have successor interactions are taken to be terminal and we attribute the immediate reward (conversions) to the corresponding (state, action) pair. Also, as is typical in the training of reinforcement learning algorithms with deep neural networks, we use a separate target network for stability \citep{mnih2015human}. We update the target network every $k=100$ training steps of the algorithm.

\begin{algorithm2e}
    \SetAlgoLined
    \SetKwInOut{Input}{Input}
    \SetKwInOut{Output}{Output}
    \nosemic \textbf{Input}: Neural network model $Q_{\theta}$, target network $Q_{\theta^-}$, initial replay buffer $\Dc$, batch-size $N$, step-size sequence $\{\alpha_k\}_{k \in \mathbb{N}}$.\;
	\textbf{Initialize}: $\theta, \theta^- \gets \theta_0$. \\
	\For{$t=0,1,\ldots$}{
		\nosemic Update replay buffer $\Dc_t$ of user transition tuples. \\
		Sample batch $\Bc=\{(s_j,a_j,r_j,s'_j,a'_j,\tau_j)\}_{j=1}^N$ from $\Dc_t$.\;
		Compute targets for each tuple:
		\[ y_j = 
		\begin{cases}
            r_j, & \text{if $s_j$ is a terminal user state} \\
            r_j + \gamma^{\tau_j}Q_{\theta^-}(s'_j,a'_j), & \text{otherwise}
        \end{cases}
        \]
		\hspace{-2.5mm} Compute loss function: $\ell(\theta_t) = \sum_{j=1}^{N} (y_j - Q_{\theta}(s_j, a_j))^2$\;
		Take a semi-gradient step: $\theta_{t+1} \gets \theta_t - \alpha_t \nabla_{\theta} \ell(\theta_t)$\;
        \If{$t \,\, \text{mod} \,\, k == 0$:}{
        Update target network:  $\theta^- \gets \theta_t$;
        }
	}
	\nosemic \textbf{Output}: estimate of state action value function $Q_{\theta}$
	\caption{Q-value estimation using SARSA}
	\label{alg: SARSA}
\end{algorithm2e}

\subsection{Data processing}
\label{subsec: data processing}
Given that recommender systems serves millions of users every day, storing interaction data for each user over long time horizons is impossible, even with distributed storage. Therefore, we take an online training approach -- collecting fresh data every day, integrating it with a fixed length table of recent interaction data for each user, and using it to generate training tuples. To describe our data processing pipeline, we introduce some notation specific to this section. We let $h$ be the \textit{effective time horizon}, specifying the number of days we look back to maintain the table of recent user interactions. Fix a time period $t > h$. Let $B_t = \left( S_{t-h}, A_{t-h}, R_{t-h}, S_{t-h+1}, \ldots, S_t, A_t, R_t \right)$ be a ``buffer'' table which stores the last $h$ days of interaction data in terms of the states, recommended items and rewards. Let $\Nc_h(t)$ be the set of users with interaction data in $B_t$. We process the agent-environment interactions at time $t+1$, $\left(S_{t+1}, A_{t+1}, R_{t+1}\right)$, into data tuples by dividing user interactions in three categories:

\begin{enumerate}[leftmargin=*]
    \item \textit{New user interactions:} Consider a user $u$ which interacts with the system at time $t+1$ but has not interacted with the system in the last $h$ time periods, i.e. $u \notin \Nc_h(t)$. Since no predecessor or successor interactions are associated with this user, we do not use this data to generate a transition tuple. Information about this interaction remains a part of the agent-environment interaction, $(S_{t+1}, A_{t+1}, R_{t+1})$, and is added to the buffer $B_{t+1}$.
    \item \textit{Active user interactions:} Consider a user $u$ which interacts with the system at time $t+1$, and this user has interacted with the system in the recent past, i.e. $u \in \Nc_{h}(t)$. Let $t_u \in \{t-h, t-h+1, \ldots, t\}$ be the most recent interaction time of user $u$ before $t+1$. Then, we take $d^u_{t+1}=(s^u_{t_u}, a^u_{t_u}, r^u_{t_u}, s^u_{t+1}, a^u_{t+1})$ to be a transition data tuple and add it to the replay buffer $\Dc$.
    \item \textit{Inactive user interactions:} Consider a user $u$ which interacted with the system at time $t-h$ but has not interacted with the system since then, including at time $t+1$. We identify such users to be inactive, and take $(s^u_{t-h}, a^u_{t-h})$ to be the terminal state action pair such that the reward $r^u_{t-h}$ is attributed to them. That is, the transition tuple $d^u_{t-h}=(s^u_{t-h}, a^u_{t-h}, r^u_{t-h}, s^*, a^*, \tau=0)$ is added to the replay buffer $\Dc$, where $(s^*, a^*)$ are dummy states and actions such that $Q_{\pibase}(s^*,a^*)=0$.
\end{enumerate}
In some way $h$ can be interpreted as the \textit{effective} time horizon in our implementation. As we work in the discounted setting and take a value of $h=15$ in our experiments, this implicitly implies\footnote{This follows as $(0.8)^{15} \approx 0.03$. We take $\gamma^{15}=0$.} a discount factor of $\gamma \approx 0.8$. Doing data processing this way with a fixed size window is only an approximation to account for very long term effects of recommendations, given data size and memory requirements of a real-time recommender system. We are currently experimenting with a different data processing pipeline architecture which can allow us to take values of $h$ up to 60.

\paragraph{Updating replay buffer and buffer table:} The new data tuples generated from agent-environment interaction $(S_{t+1}, A_{t+1}, R_{t+1})$ are added to the replay buffer $\Dc$. We also update the buffer table $B_{t+1}=(S_{t+1-h}, A_{t+1-h}, R_{t+1-h}, \ldots, S_{t+1}, A_{t+1}, R_{t+1})$ to store the agent environment interactions of past $h$ periods at time $t+1$. This constitutes the ``update replay buffer'' step in Algorithm \ref{alg: SARSA}. See Figure \ref{fig: data processing} below for an illustration of different types of user interaction tuples generated by the agent-environment interaction $(S_{t+1}, A_{t+1}, R_{t+1})$ which are used to update the buffer table $B_{t+1}$.
\vspace{5pt}

\begin{figure}[ht]
  \centering
  \includegraphics[width=0.9\linewidth]{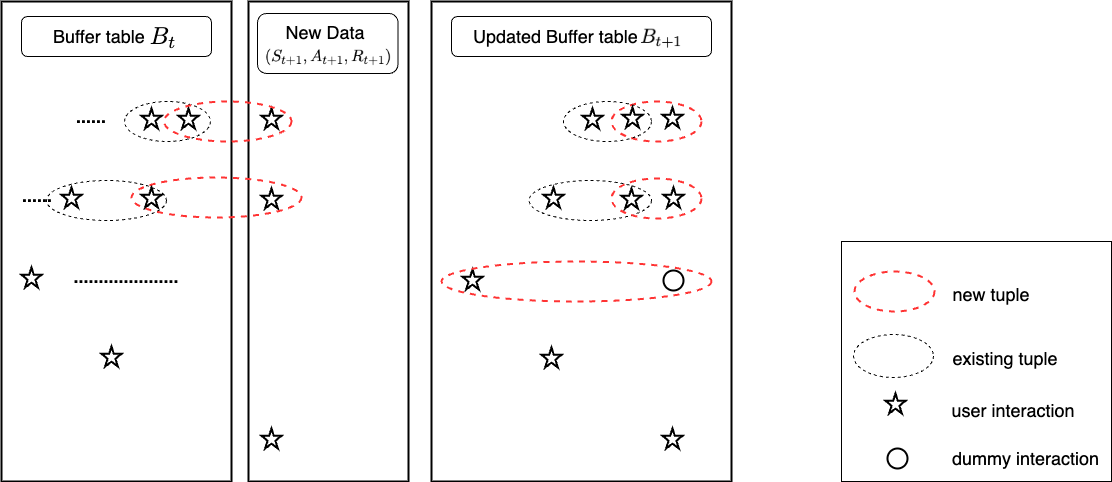}
  \caption{Data processing pipeline with different user interactions. First two rows show how \textit{new data interaction} tuples are created and updated in the buffer table. The third row illustrates \textit{inactive user interaction}, where a dummy state-action is appended to create a data tuple. Fourth row shows that recent interactions are carried over in the updated buffer table $B_{t+1}$ for users who did not interact with the system at time $t+1$. The final row illustrates a \textit{new user interaction}.}
    \label{fig: data processing}
\end{figure}

\vspace{-10pt}
\section{Experiments}
\label{sec: experiments}

\begin{figure*}
\begin{minipage}{.5\linewidth}
\centering
\subfloat{\includegraphics[width=0.95\linewidth]{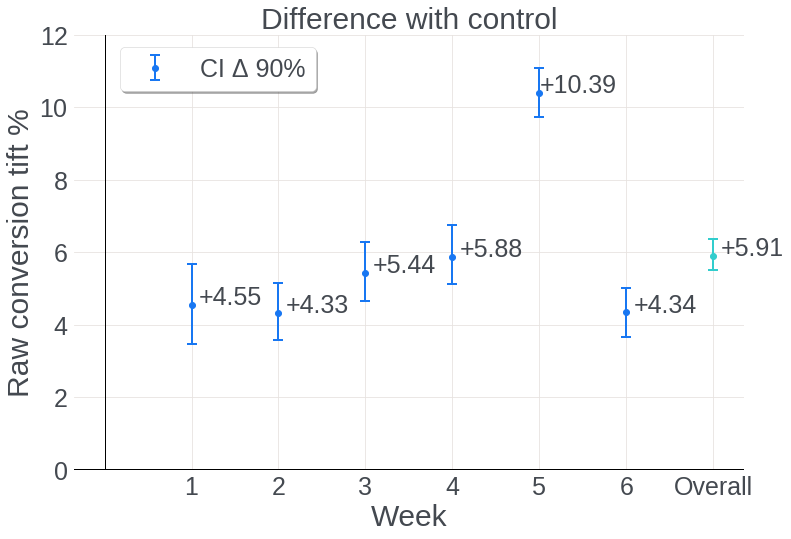}}
\end{minipage}%
\begin{minipage}{.5\linewidth}
\centering
\subfloat{\includegraphics[width=0.95\linewidth]{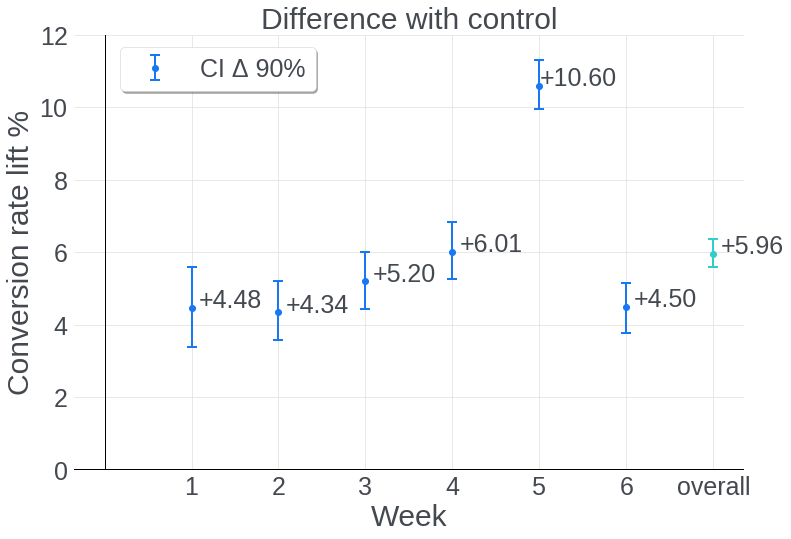}}
\end{minipage}
\caption{Weekly change in average number of conversions and conversion rate over six weeks of experimentation. Note that change in raw conversion numbers closely mimics the change in conversion rates as we tried to maintain similar total impressions over the test and control group. Error bars represent the standard deviation of change across different bidders.}
\label{fig: weekly change}
\par\medskip
\vspace{10pt}
\centering
\subfloat{\includegraphics[height=4cm, width=0.99\linewidth]{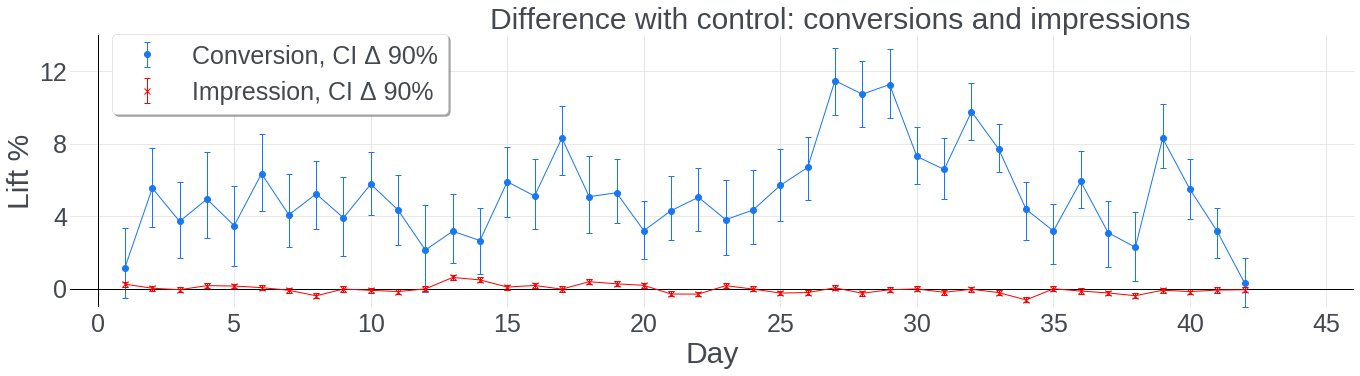}}
\caption{Daily change in conversion and impression numbers over the control group. For most days, our approach (blue curve) results in a significant increase in conversions over the auction based recommender policy. The red curve shows the change in number of daily impressions, which approximately stayed close to 0 for most days.}
\label{fig: daily change}
\end{figure*}
We conduct an online A/B experiment over a period of six weeks to test our proposed RL based approach in a web scale recommender system serving billions of users every day. The base policy in our experiment uses a second price auction mechanism to recommend items to users. On average, about 2 billion users interact with our recommender system every day. For our experiments, we randomly selected 2\% (40 million users) to be a part of our test and control group each. Recall from section \ref{subsec: data processing} that our estimated Q values account for cumulative conversions over a 14 day period. 

To stay close to the base policy, we tune the hyperparameter $\alpha$ introduced in section \ref{ltv_weighted_sum} to make sure that the contribution of estimated cumulative conversions, $\alpha Q_{\pibase}(s,a)$, in the scoring model $(1-\alpha) f(s,a) + \alpha Q_{\pibase}(s,a)$ is less than 8\%. That is, we make sure that the fraction $\alpha Q_{\pibase}(s,a)/\left( (1-\alpha) f(s,a) + \alpha Q_{\pibase}(s,a) \right) \leq 0.08$. After hyperparameter tuning or our experiments, this approximately came to a value of $\alpha=0.96$.




\paragraph{Implementation details:} As described in section \ref{sec: problem formulation}, we model the user state as a tuple of three features, $s^u = (z^u, x^u, i^u)$. The context features $x^u$ which encode user demographics are extracted respecting privacy constraints. A sequence model is used to generate features $z^u$ which summarizes the user interests by encoding their past interactions with the system while graph based models pretrained on a large-scale internal dataset (which goes beyond interactions of users with the recommender system) are used to extract features $i^u$ encoding user preferences. The item embeddings, denoted by $a$, include textual and image representations along with its content category. We use a fully connected deep neural network to model the Q-values $Q_{\pibase}(s^u, a)$, by taking state and action representations as inputs. We retrain our models every 12 hours to adapt to data distribution shifts and maintain a replay buffer of approximately 300 million data tuples at anytime to train our models using SARSA.

\paragraph{Experimental results:} Figures \ref{fig: weekly change} and \ref{fig: daily change} summarize the results of our A/B test. In Figure \ref{fig: weekly change}, we aggregate results week-by-week and summarize the overall improvement for both the number of conversions as well as the conversion rate. On average we see a lift (change between test and control group) between $4-10\%$ for both metrics which shows a significant improvement over the base policy. Moreover, Fig. \ref{fig: daily change} shows that improvements of our approach hold consistently, over most days. Overall, our results demonstrate the effectiveness of our approach in optimizing for long run metrics. We remark that change in conversions closely mimics the change in conversion rates as we tried to maintain a neutral impression change. This can be seen in Fig. \ref{fig: daily change} where the change in impressions for the test group over control is close to zero for all days.

\section{Conclusion}
In this work, we proposed a reinforcement learning based approach to account for long run user engagement in recommending items to users while still working within the constraints of an auction based recommender system. Foundational theory in dynamic programming motivates our method -- we show that our method provably improves over the baseline policy. An online A/B test conducted with a billion-user scale recommender system demonstrates significant improvements and illustrates the potential of using reinforcement learning ideas for recommender systems at scale in production.

\newpage
\bibliographystyle{ACM-Reference-Format}
\bibliography{LTV_RecSys}


\begin{thebibliography}{42}


\ifx \showCODEN    \undefined \def \showCODEN     #1{\unskip}     \fi
\ifx \showDOI      \undefined \def \showDOI       #1{#1}\fi
\ifx \showISBNx    \undefined \def \showISBNx     #1{\unskip}     \fi
\ifx \showISBNxiii \undefined \def \showISBNxiii  #1{\unskip}     \fi
\ifx \showISSN     \undefined \def \showISSN      #1{\unskip}     \fi
\ifx \showLCCN     \undefined \def \showLCCN      #1{\unskip}     \fi
\ifx \shownote     \undefined \def \shownote      #1{#1}          \fi
\ifx \showarticletitle \undefined \def \showarticletitle #1{#1}   \fi
\ifx \showURL      \undefined \def \showURL       {\relax}        \fi
\providecommand\bibfield[2]{#2}
\providecommand\bibinfo[2]{#2}
\providecommand\natexlab[1]{#1}
\providecommand\showeprint[2][]{arXiv:#2}

\bibitem[Bertsekas(2012)]%
        {bertsekas2012dynamic}
\bibfield{author}{\bibinfo{person}{Dimitri Bertsekas}.}
  \bibinfo{year}{2012}\natexlab{}.
\newblock \bibinfo{booktitle}{\emph{Dynamic programming and optimal control:
  Volume I}}. Vol.~\bibinfo{volume}{1}.
\newblock \bibinfo{publisher}{Athena scientific}.
\newblock


\bibitem[Chen et~al\mbox{.}(2019a)]%
        {chen2019top}
\bibfield{author}{\bibinfo{person}{Minmin Chen}, \bibinfo{person}{Alex Beutel},
  \bibinfo{person}{Paul Covington}, \bibinfo{person}{Sagar Jain},
  \bibinfo{person}{Francois Belletti}, {and} \bibinfo{person}{Ed~H Chi}.}
  \bibinfo{year}{2019}\natexlab{a}.
\newblock \showarticletitle{Top-k off-policy correction for a REINFORCE
  recommender system}. In \bibinfo{booktitle}{\emph{Proceedings of the Twelfth
  ACM International Conference on Web Search and Data Mining}}.
  \bibinfo{pages}{456--464}.
\newblock


\bibitem[Chen et~al\mbox{.}(2022)]%
        {chen2022off}
\bibfield{author}{\bibinfo{person}{Minmin Chen}, \bibinfo{person}{Can Xu},
  \bibinfo{person}{Vince Gatto}, \bibinfo{person}{Devanshu Jain},
  \bibinfo{person}{Aviral Kumar}, {and} \bibinfo{person}{Ed Chi}.}
  \bibinfo{year}{2022}\natexlab{}.
\newblock \showarticletitle{Off-Policy Actor-critic for Recommender Systems}.
  In \bibinfo{booktitle}{\emph{Proceedings of the 16th ACM Conference on
  Recommender Systems}}. \bibinfo{pages}{338--349}.
\newblock


\bibitem[Chen et~al\mbox{.}(2019b)]%
        {chen2019generative}
\bibfield{author}{\bibinfo{person}{Xinshi Chen}, \bibinfo{person}{Shuang Li},
  \bibinfo{person}{Hui Li}, \bibinfo{person}{Shaohua Jiang},
  \bibinfo{person}{Yuan Qi}, {and} \bibinfo{person}{Le Song}.}
  \bibinfo{year}{2019}\natexlab{b}.
\newblock \showarticletitle{Generative adversarial user model for reinforcement
  learning based recommendation system}. In
  \bibinfo{booktitle}{\emph{International Conference on Machine Learning}}.
  PMLR, \bibinfo{pages}{1052--1061}.
\newblock


\bibitem[Cheng et~al\mbox{.}(2016)]%
        {cheng2016wide}
\bibfield{author}{\bibinfo{person}{Heng-Tze Cheng}, \bibinfo{person}{Levent
  Koc}, \bibinfo{person}{Jeremiah Harmsen}, \bibinfo{person}{Tal Shaked},
  \bibinfo{person}{Tushar Chandra}, \bibinfo{person}{Hrishi Aradhye},
  \bibinfo{person}{Glen Anderson}, \bibinfo{person}{Greg Corrado},
  \bibinfo{person}{Wei Chai}, \bibinfo{person}{Mustafa Ispir}, {et~al\mbox{.}}}
  \bibinfo{year}{2016}\natexlab{}.
\newblock \showarticletitle{Wide \& deep learning for recommender systems}. In
  \bibinfo{booktitle}{\emph{Proceedings of the 1st workshop on deep learning
  for recommender systems}}. \bibinfo{pages}{7--10}.
\newblock


\bibitem[Covington et~al\mbox{.}(2016)]%
        {covington2016deep}
\bibfield{author}{\bibinfo{person}{Paul Covington}, \bibinfo{person}{Jay
  Adams}, {and} \bibinfo{person}{Emre Sargin}.}
  \bibinfo{year}{2016}\natexlab{}.
\newblock \showarticletitle{Deep neural networks for youtube recommendations}.
  In \bibinfo{booktitle}{\emph{Proceedings of the 10th ACM conference on
  recommender systems}}. \bibinfo{pages}{191--198}.
\newblock


\bibitem[Edelman et~al\mbox{.}(2007)]%
        {edelman2007internet}
\bibfield{author}{\bibinfo{person}{Benjamin Edelman}, \bibinfo{person}{Michael
  Ostrovsky}, {and} \bibinfo{person}{Michael Schwarz}.}
  \bibinfo{year}{2007}\natexlab{}.
\newblock \showarticletitle{Internet advertising and the generalized
  second-price auction: Selling billions of dollars worth of keywords}.
\newblock \bibinfo{journal}{\emph{American economic review}}
  \bibinfo{volume}{97}, \bibinfo{number}{1} (\bibinfo{year}{2007}),
  \bibinfo{pages}{242--259}.
\newblock


\bibitem[Ekstrand et~al\mbox{.}(2011)]%
        {ekstrand2011collaborative}
\bibfield{author}{\bibinfo{person}{Michael~D Ekstrand}, \bibinfo{person}{John~T
  Riedl}, \bibinfo{person}{Joseph~A Konstan}, {et~al\mbox{.}}}
  \bibinfo{year}{2011}\natexlab{}.
\newblock \showarticletitle{Collaborative filtering recommender systems}.
\newblock \bibinfo{journal}{\emph{Foundations and Trends{\textregistered} in
  Human--Computer Interaction}} \bibinfo{volume}{4}, \bibinfo{number}{2}
  (\bibinfo{year}{2011}), \bibinfo{pages}{81--173}.
\newblock


\bibitem[Evans(2008)]%
        {evans2008economics}
\bibfield{author}{\bibinfo{person}{David~S Evans}.}
  \bibinfo{year}{2008}\natexlab{}.
\newblock \showarticletitle{The economics of the online advertising industry}.
\newblock \bibinfo{journal}{\emph{Review of network economics}}
  \bibinfo{volume}{7}, \bibinfo{number}{3} (\bibinfo{year}{2008}).
\newblock


\bibitem[Gomez-Uribe and Hunt(2015)]%
        {gomez2015netflix}
\bibfield{author}{\bibinfo{person}{Carlos~A Gomez-Uribe} {and}
  \bibinfo{person}{Neil Hunt}.} \bibinfo{year}{2015}\natexlab{}.
\newblock \showarticletitle{The netflix recommender system: Algorithms,
  business value, and innovation}.
\newblock \bibinfo{journal}{\emph{ACM Transactions on Management Information
  Systems (TMIS)}} \bibinfo{volume}{6}, \bibinfo{number}{4}
  (\bibinfo{year}{2015}), \bibinfo{pages}{1--19}.
\newblock


\bibitem[Gu et~al\mbox{.}(2017)]%
        {gu2017deep}
\bibfield{author}{\bibinfo{person}{Shixiang Gu}, \bibinfo{person}{Ethan Holly},
  \bibinfo{person}{Timothy Lillicrap}, {and} \bibinfo{person}{Sergey Levine}.}
  \bibinfo{year}{2017}\natexlab{}.
\newblock \showarticletitle{Deep reinforcement learning for robotic
  manipulation with asynchronous off-policy updates}. In
  \bibinfo{booktitle}{\emph{2017 IEEE international conference on robotics and
  automation (ICRA)}}. IEEE, \bibinfo{pages}{3389--3396}.
\newblock


\bibitem[Guo et~al\mbox{.}(2023)]%
        {Guo2023evaluate}
\bibfield{author}{\bibinfo{person}{Hongbo Guo}, \bibinfo{person}{Ruben Naeff},
  \bibinfo{person}{Alex Nikulkov}, {and} \bibinfo{person}{Zheqing Zhu}.}
  \bibinfo{year}{2023}\natexlab{}.
\newblock \showarticletitle{Evaluating Online Bandit Exploration In Large-Scale
  Recommender System}. In \bibinfo{booktitle}{\emph{KDD-23 Workshop on
  Multi-Armed Bandits and Reinforcement Learning: Advancing Decision Making in
  E-Commerce and Beyond}}.
\newblock


\bibitem[Hidasi et~al\mbox{.}(2015)]%
        {hidasi2015session}
\bibfield{author}{\bibinfo{person}{Bal{\'a}zs Hidasi},
  \bibinfo{person}{Alexandros Karatzoglou}, \bibinfo{person}{Linas Baltrunas},
  {and} \bibinfo{person}{Domonkos Tikk}.} \bibinfo{year}{2015}\natexlab{}.
\newblock \showarticletitle{Session-based recommendations with recurrent neural
  networks}.
\newblock \bibinfo{journal}{\emph{arXiv preprint arXiv:1511.06939}}
  (\bibinfo{year}{2015}).
\newblock


\bibitem[Hutter(2007)]%
        {hutter2007universal}
\bibfield{author}{\bibinfo{person}{Marcus Hutter}.}
  \bibinfo{year}{2007}\natexlab{}.
\newblock \showarticletitle{Universal algorithmic intelligence: A mathematical
  top→ down approach}.
\newblock \bibinfo{journal}{\emph{Artificial general intelligence}}
  (\bibinfo{year}{2007}), \bibinfo{pages}{227--290}.
\newblock


\bibitem[Ie et~al\mbox{.}(2019)]%
        {ie2019slateq}
\bibfield{author}{\bibinfo{person}{Eugene Ie}, \bibinfo{person}{Vihan Jain},
  \bibinfo{person}{Jing Wang}, \bibinfo{person}{Sanmit Narvekar},
  \bibinfo{person}{Ritesh Agarwal}, \bibinfo{person}{Rui Wu},
  \bibinfo{person}{Heng-Tze Cheng}, \bibinfo{person}{Tushar Chandra}, {and}
  \bibinfo{person}{Craig Boutilier}.} \bibinfo{year}{2019}\natexlab{}.
\newblock \showarticletitle{SlateQ: A tractable decomposition for reinforcement
  learning with recommendation sets}.
\newblock  (\bibinfo{year}{2019}).
\newblock


\bibitem[Koren et~al\mbox{.}(2009)]%
        {koren2009matrix}
\bibfield{author}{\bibinfo{person}{Yehuda Koren}, \bibinfo{person}{Robert
  Bell}, {and} \bibinfo{person}{Chris Volinsky}.}
  \bibinfo{year}{2009}\natexlab{}.
\newblock \showarticletitle{Matrix factorization techniques for recommender
  systems}.
\newblock \bibinfo{journal}{\emph{Computer}} \bibinfo{volume}{42},
  \bibinfo{number}{8} (\bibinfo{year}{2009}), \bibinfo{pages}{30--37}.
\newblock


\bibitem[Li et~al\mbox{.}(2010)]%
        {li2010contextual}
\bibfield{author}{\bibinfo{person}{Lihong Li}, \bibinfo{person}{Wei Chu},
  \bibinfo{person}{John Langford}, {and} \bibinfo{person}{Robert~E Schapire}.}
  \bibinfo{year}{2010}\natexlab{}.
\newblock \showarticletitle{A contextual-bandit approach to personalized news
  article recommendation}. In \bibinfo{booktitle}{\emph{Proceedings of the 19th
  international conference on World wide web}}. \bibinfo{pages}{661--670}.
\newblock


\bibitem[Lops et~al\mbox{.}(2011)]%
        {lops2011content}
\bibfield{author}{\bibinfo{person}{Pasquale Lops}, \bibinfo{person}{Marco
  De~Gemmis}, {and} \bibinfo{person}{Giovanni Semeraro}.}
  \bibinfo{year}{2011}\natexlab{}.
\newblock \showarticletitle{Content-based recommender systems: State of the art
  and trends}.
\newblock \bibinfo{journal}{\emph{Recommender systems handbook}}
  (\bibinfo{year}{2011}), \bibinfo{pages}{73--105}.
\newblock


\bibitem[Lu et~al\mbox{.}(2015)]%
        {lu2015recommender}
\bibfield{author}{\bibinfo{person}{Jie Lu}, \bibinfo{person}{Dianshuang Wu},
  \bibinfo{person}{Mingsong Mao}, \bibinfo{person}{Wei Wang}, {and}
  \bibinfo{person}{Guangquan Zhang}.} \bibinfo{year}{2015}\natexlab{}.
\newblock \showarticletitle{Recommender system application developments: a
  survey}.
\newblock \bibinfo{journal}{\emph{Decision support systems}}
  \bibinfo{volume}{74} (\bibinfo{year}{2015}), \bibinfo{pages}{12--32}.
\newblock


\bibitem[Lu et~al\mbox{.}(2021)]%
        {lu2021reinforcement}
\bibfield{author}{\bibinfo{person}{Xiuyuan Lu}, \bibinfo{person}{Benjamin
  Van~Roy}, \bibinfo{person}{Vikranth Dwaracherla}, \bibinfo{person}{Morteza
  Ibrahimi}, \bibinfo{person}{Ian Osband}, {and} \bibinfo{person}{Zheng Wen}.}
  \bibinfo{year}{2021}\natexlab{}.
\newblock \showarticletitle{Reinforcement learning, bit by bit}.
\newblock \bibinfo{journal}{\emph{arXiv preprint arXiv:2103.04047}}
  (\bibinfo{year}{2021}).
\newblock


\bibitem[Mnih et~al\mbox{.}(2015)]%
        {mnih2015human}
\bibfield{author}{\bibinfo{person}{Volodymyr Mnih}, \bibinfo{person}{Koray
  Kavukcuoglu}, \bibinfo{person}{David Silver}, \bibinfo{person}{Andrei~A
  Rusu}, \bibinfo{person}{Joel Veness}, \bibinfo{person}{Marc~G Bellemare},
  \bibinfo{person}{Alex Graves}, \bibinfo{person}{Martin Riedmiller},
  \bibinfo{person}{Andreas~K Fidjeland}, \bibinfo{person}{Georg Ostrovski},
  {et~al\mbox{.}}} \bibinfo{year}{2015}\natexlab{}.
\newblock \showarticletitle{Human-level control through deep reinforcement
  learning}.
\newblock \bibinfo{journal}{\emph{nature}} \bibinfo{volume}{518},
  \bibinfo{number}{7540} (\bibinfo{year}{2015}), \bibinfo{pages}{529--533}.
\newblock


\bibitem[Okura et~al\mbox{.}(2017)]%
        {okura2017embedding}
\bibfield{author}{\bibinfo{person}{Shumpei Okura}, \bibinfo{person}{Yukihiro
  Tagami}, \bibinfo{person}{Shingo Ono}, {and} \bibinfo{person}{Akira Tajima}.}
  \bibinfo{year}{2017}\natexlab{}.
\newblock \showarticletitle{Embedding-based news recommendation for millions of
  users}. In \bibinfo{booktitle}{\emph{Proceedings of the 23rd ACM SIGKDD
  international conference on knowledge discovery and data mining}}.
  \bibinfo{pages}{1933--1942}.
\newblock


\bibitem[Pazzani and Billsus(2007)]%
        {pazzani2007content}
\bibfield{author}{\bibinfo{person}{Michael~J Pazzani} {and}
  \bibinfo{person}{Daniel Billsus}.} \bibinfo{year}{2007}\natexlab{}.
\newblock \showarticletitle{Content-based recommendation systems}.
\newblock \bibinfo{journal}{\emph{The adaptive web: methods and strategies of
  web personalization}} (\bibinfo{year}{2007}), \bibinfo{pages}{325--341}.
\newblock


\bibitem[Qin et~al\mbox{.}(2014)]%
        {qin2014contextual}
\bibfield{author}{\bibinfo{person}{Lijing Qin}, \bibinfo{person}{Shouyuan
  Chen}, {and} \bibinfo{person}{Xiaoyan Zhu}.} \bibinfo{year}{2014}\natexlab{}.
\newblock \showarticletitle{Contextual combinatorial bandit and its application
  on diversified online recommendation}. In
  \bibinfo{booktitle}{\emph{Proceedings of the 2014 SIAM International
  Conference on Data Mining}}. SIAM, \bibinfo{pages}{461--469}.
\newblock


\bibitem[Quadrana et~al\mbox{.}(2018)]%
        {quadrana2018sequence}
\bibfield{author}{\bibinfo{person}{Massimo Quadrana}, \bibinfo{person}{Paolo
  Cremonesi}, {and} \bibinfo{person}{Dietmar Jannach}.}
  \bibinfo{year}{2018}\natexlab{}.
\newblock \showarticletitle{Sequence-aware recommender systems}.
\newblock \bibinfo{journal}{\emph{ACM Computing Surveys (CSUR)}}
  \bibinfo{volume}{51}, \bibinfo{number}{4} (\bibinfo{year}{2018}),
  \bibinfo{pages}{1--36}.
\newblock


\bibitem[Schafer et~al\mbox{.}(2007)]%
        {schafer2007collaborative}
\bibfield{author}{\bibinfo{person}{J~Ben Schafer}, \bibinfo{person}{Dan
  Frankowski}, \bibinfo{person}{Jon Herlocker}, {and} \bibinfo{person}{Shilad
  Sen}.} \bibinfo{year}{2007}\natexlab{}.
\newblock \showarticletitle{Collaborative filtering recommender systems}.
\newblock \bibinfo{journal}{\emph{The adaptive web: methods and strategies of
  web personalization}} (\bibinfo{year}{2007}), \bibinfo{pages}{291--324}.
\newblock


\bibitem[Shani et~al\mbox{.}(2005)]%
        {shani2005mdp}
\bibfield{author}{\bibinfo{person}{Guy Shani}, \bibinfo{person}{David
  Heckerman}, \bibinfo{person}{Ronen~I Brafman}, {and} \bibinfo{person}{Craig
  Boutilier}.} \bibinfo{year}{2005}\natexlab{}.
\newblock \showarticletitle{An MDP-based recommender system.}
\newblock \bibinfo{journal}{\emph{Journal of Machine Learning Research}}
  \bibinfo{volume}{6}, \bibinfo{number}{9} (\bibinfo{year}{2005}).
\newblock


\bibitem[Shi et~al\mbox{.}(2014)]%
        {shi2014collaborative}
\bibfield{author}{\bibinfo{person}{Yue Shi}, \bibinfo{person}{Martha Larson},
  {and} \bibinfo{person}{Alan Hanjalic}.} \bibinfo{year}{2014}\natexlab{}.
\newblock \showarticletitle{Collaborative filtering beyond the user-item
  matrix: A survey of the state of the art and future challenges}.
\newblock \bibinfo{journal}{\emph{ACM Computing Surveys (CSUR)}}
  \bibinfo{volume}{47}, \bibinfo{number}{1} (\bibinfo{year}{2014}),
  \bibinfo{pages}{1--45}.
\newblock


\bibitem[Smith and Linden(2017)]%
        {smith2017two}
\bibfield{author}{\bibinfo{person}{Brent Smith} {and} \bibinfo{person}{Greg
  Linden}.} \bibinfo{year}{2017}\natexlab{}.
\newblock \showarticletitle{Two decades of recommender systems at Amazon. com}.
\newblock \bibinfo{journal}{\emph{Ieee internet computing}}
  \bibinfo{volume}{21}, \bibinfo{number}{3} (\bibinfo{year}{2017}),
  \bibinfo{pages}{12--18}.
\newblock


\bibitem[Tai et~al\mbox{.}(2017)]%
        {tai2017virtual}
\bibfield{author}{\bibinfo{person}{Lei Tai}, \bibinfo{person}{Giuseppe Paolo},
  {and} \bibinfo{person}{Ming Liu}.} \bibinfo{year}{2017}\natexlab{}.
\newblock \showarticletitle{Virtual-to-real deep reinforcement learning:
  Continuous control of mobile robots for mapless navigation}. In
  \bibinfo{booktitle}{\emph{2017 IEEE/RSJ International Conference on
  Intelligent Robots and Systems (IROS)}}. IEEE, \bibinfo{pages}{31--36}.
\newblock


\bibitem[Varian(2007)]%
        {varian2007position}
\bibfield{author}{\bibinfo{person}{Hal~R Varian}.}
  \bibinfo{year}{2007}\natexlab{}.
\newblock \showarticletitle{Position auctions}.
\newblock \bibinfo{journal}{\emph{international Journal of industrial
  Organization}} \bibinfo{volume}{25}, \bibinfo{number}{6}
  (\bibinfo{year}{2007}), \bibinfo{pages}{1163--1178}.
\newblock


\bibitem[Varian and Harris(2014)]%
        {varian2014vcg}
\bibfield{author}{\bibinfo{person}{Hal~R Varian} {and}
  \bibinfo{person}{Christopher Harris}.} \bibinfo{year}{2014}\natexlab{}.
\newblock \showarticletitle{The VCG auction in theory and practice}.
\newblock \bibinfo{journal}{\emph{American Economic Review}}
  \bibinfo{volume}{104}, \bibinfo{number}{5} (\bibinfo{year}{2014}),
  \bibinfo{pages}{442--445}.
\newblock


\bibitem[Wang et~al\mbox{.}(2017)]%
        {wang2017factorization}
\bibfield{author}{\bibinfo{person}{Huazheng Wang}, \bibinfo{person}{Qingyun
  Wu}, {and} \bibinfo{person}{Hongning Wang}.} \bibinfo{year}{2017}\natexlab{}.
\newblock \showarticletitle{Factorization bandits for interactive
  recommendation}. In \bibinfo{booktitle}{\emph{Proceedings of the AAAI
  Conference on Artificial Intelligence}}, Vol.~\bibinfo{volume}{31}.
\newblock


\bibitem[Wang et~al\mbox{.}(2019)]%
        {wang2019sequential}
\bibfield{author}{\bibinfo{person}{Shoujin Wang}, \bibinfo{person}{Liang Hu},
  \bibinfo{person}{Yan Wang}, \bibinfo{person}{Longbing Cao},
  \bibinfo{person}{Quan~Z Sheng}, {and} \bibinfo{person}{Mehmet Orgun}.}
  \bibinfo{year}{2019}\natexlab{}.
\newblock \showarticletitle{Sequential recommender systems: challenges,
  progress and prospects}.
\newblock \bibinfo{journal}{\emph{arXiv preprint arXiv:2001.04830}}
  (\bibinfo{year}{2019}).
\newblock


\bibitem[Xin et~al\mbox{.}(2020)]%
        {xin2020self}
\bibfield{author}{\bibinfo{person}{Xin Xin}, \bibinfo{person}{Alexandros
  Karatzoglou}, \bibinfo{person}{Ioannis Arapakis}, {and}
  \bibinfo{person}{Joemon~M Jose}.} \bibinfo{year}{2020}\natexlab{}.
\newblock \showarticletitle{Self-supervised reinforcement learning for
  recommender systems}. In \bibinfo{booktitle}{\emph{Proceedings of the 43rd
  International ACM SIGIR conference on research and development in Information
  Retrieval}}. \bibinfo{pages}{931--940}.
\newblock


\bibitem[Zeng et~al\mbox{.}(2016)]%
        {zeng2016online}
\bibfield{author}{\bibinfo{person}{Chunqiu Zeng}, \bibinfo{person}{Qing Wang},
  \bibinfo{person}{Shekoofeh Mokhtari}, {and} \bibinfo{person}{Tao Li}.}
  \bibinfo{year}{2016}\natexlab{}.
\newblock \showarticletitle{Online context-aware recommendation with time
  varying multi-armed bandit}. In \bibinfo{booktitle}{\emph{Proceedings of the
  22nd ACM SIGKDD international conference on Knowledge discovery and data
  mining}}. \bibinfo{pages}{2025--2034}.
\newblock


\bibitem[Zhang et~al\mbox{.}(2019)]%
        {zhang2019deep}
\bibfield{author}{\bibinfo{person}{Shuai Zhang}, \bibinfo{person}{Lina Yao},
  \bibinfo{person}{Aixin Sun}, {and} \bibinfo{person}{Yi Tay}.}
  \bibinfo{year}{2019}\natexlab{}.
\newblock \showarticletitle{Deep learning based recommender system: A survey
  and new perspectives}.
\newblock \bibinfo{journal}{\emph{ACM computing surveys (CSUR)}}
  \bibinfo{volume}{52}, \bibinfo{number}{1} (\bibinfo{year}{2019}),
  \bibinfo{pages}{1--38}.
\newblock


\bibitem[Zhao et~al\mbox{.}(2019)]%
        {zhao2019deep}
\bibfield{author}{\bibinfo{person}{Xiangyu Zhao}, \bibinfo{person}{Long Xia},
  \bibinfo{person}{Jiliang Tang}, {and} \bibinfo{person}{Dawei Yin}.}
  \bibinfo{year}{2019}\natexlab{}.
\newblock \showarticletitle{" Deep reinforcement learning for search,
  recommendation, and online advertising: a survey" by Xiangyu Zhao, Long Xia,
  Jiliang Tang, and Dawei Yin with Martin Vesely as coordinator}.
\newblock \bibinfo{journal}{\emph{ACM Sigweb Newsletter}}
  \bibinfo{number}{Spring} (\bibinfo{year}{2019}), \bibinfo{pages}{1--15}.
\newblock


\bibitem[Zheng et~al\mbox{.}(2018)]%
        {zheng2018drn}
\bibfield{author}{\bibinfo{person}{Guanjie Zheng}, \bibinfo{person}{Fuzheng
  Zhang}, \bibinfo{person}{Zihan Zheng}, \bibinfo{person}{Yang Xiang},
  \bibinfo{person}{Nicholas~Jing Yuan}, \bibinfo{person}{Xing Xie}, {and}
  \bibinfo{person}{Zhenhui Li}.} \bibinfo{year}{2018}\natexlab{}.
\newblock \showarticletitle{DRN: A deep reinforcement learning framework for
  news recommendation}. In \bibinfo{booktitle}{\emph{Proceedings of the 2018
  world wide web conference}}. \bibinfo{pages}{167--176}.
\newblock


\bibitem[Zhu and Van~Roy(2023a)]%
        {zhu2023deep}
\bibfield{author}{\bibinfo{person}{Zheqing Zhu} {and} \bibinfo{person}{Benjamin
  Van~Roy}.} \bibinfo{year}{2023}\natexlab{a}.
\newblock \showarticletitle{Deep Exploration for Recommendation Systems}. In
  \bibinfo{booktitle}{\emph{Proceedings of the 17th ACM Conference on
  Recommender Systems}}.
\newblock


\bibitem[Zhu and Van~Roy(2023b)]%
        {zhu2023scalable}
\bibfield{author}{\bibinfo{person}{Zheqing Zhu} {and} \bibinfo{person}{Benjamin
  Van~Roy}.} \bibinfo{year}{2023}\natexlab{b}.
\newblock \showarticletitle{Scalable Neural Contextual Bandit for Recommender
  Systems}.
\newblock \bibinfo{journal}{\emph{arXiv preprint arXiv:2306.14834}}
  (\bibinfo{year}{2023}).
\newblock


\bibitem[Zou et~al\mbox{.}(2019)]%
        {zou2019reinforcement}
\bibfield{author}{\bibinfo{person}{Lixin Zou}, \bibinfo{person}{Long Xia},
  \bibinfo{person}{Zhuoye Ding}, \bibinfo{person}{Jiaxing Song},
  \bibinfo{person}{Weidong Liu}, {and} \bibinfo{person}{Dawei Yin}.}
  \bibinfo{year}{2019}\natexlab{}.
\newblock \showarticletitle{Reinforcement learning to optimize long-term user
  engagement in recommender systems}. In \bibinfo{booktitle}{\emph{Proceedings
  of the 25th ACM SIGKDD International Conference on Knowledge Discovery \&
  Data Mining}}. \bibinfo{pages}{2810--2818}.
\newblock


\end{thebibliography}

\appendix

\end{document}